\documentclass{llncs}
\usepackage{epsfig,color}
\usepackage{graphicx,amsfonts}
\usepackage{algorithm}
\usepackage{algorithmic}

\newtheorem{observation}{Observation}

\title{Unit Disk Cover Problem}
\author{Rashmisnata Acharyya\inst{1}, Manjanna B\inst{2} \and Gautam K. Das\inst{2}}
\institute{Department of Computer Science and Engineering,\\ Tezpur University,
Assam, 786028, India
 \and
Department of Mathematics,\\Indian Institute of Technology, Guwahatii, 781 - 039, India}
\date{}
\begin{document}
\maketitle
\input{psfig.sty}

\begin{abstract}
Given a set ${\cal D}$ of unit disks in the Euclidean plane, we consider (i) the
{\it discrete unit disk cover} (DUDC) problem and (ii) the {\it rectangular
region cover} (RRC) problem. In the DUDC problem, for a given set ${\cal P}$
of points the objective is to select minimum cardinality subset
${\cal D}^* \subseteq {\cal D}$ such that each point in ${\cal P}$ is covered by at
least one disk in ${\cal D}^*$. On the other hand, in the RRC problem the objective
is to select minimum cardinality subset ${\cal D}^{**} \subseteq {\cal D}$ such that
each point of a given rectangular region ${\cal R}$ is covered by a disk in
${\cal D}^{**}$. For the DUDC problem, we propose an $(9+\epsilon)$-factor
($0 < \epsilon \leq 6$) approximation algorithm. The previous best known approximation 
factor was 15 \cite{FL12}. For the RRC problem, we propose (i) an $(9 + \epsilon)$-factor
($0 < \epsilon \leq 6$) approximation algorithm, (ii) an 2.25-factor approximation
algorithm in reduce radius setup, improving previous 4-factor approximation result in the
same setup \cite{FKKLS07}.

The solution of DUDC problem is based on a PTAS for the subproblem LSDUDC, where all
the points in ${\cal P}$ are on one side of a line and covered by the disks centered on
the other side of that line.
\end{abstract}

\section{Introduction}
In the {\it unit disk cover} (UDC) problem, we consider two problems, namely
the {\it discrete unit disk cover} (DUDC) problem and the {\it rectangular region cover}
(RRC) problem. In the DUDC problem, given a set ${\cal P} = \{p_1, p_2, \ldots, p_n\}$
of $n$ points and a set ${\cal D} = \{d_1, d_2, \ldots, d_m\}$ of $m$ unit disks
in the Euclidean plane, we wish to determine the minimum cardinality set
${\cal D}^* \subseteq {\cal D}$ such that ${\cal P} \cap {\cal D}^* = {\cal P}$.
In the {\it rectangular region cover} (RRC) problem, given a rectangular region ${\cal R}$
and a set ${\cal D} = \{d_1, d_2, \ldots, d_m\}$ of $m$ unit disks in the Euclidean plane,
the objective is to determine the minimum cardinality set ${\cal D}^{**} \subseteq {\cal D}$
such that ${\cal R} \cap {\cal D}^{**} = {\cal R}$. The DUDC and RRC problems are a
geometric version of the general set cover problem which is known to be NP-complete \cite{GJ79}.
The general set cover problem is not approximable within $c \log n$, for some constant $c$,
where $n$ is the size of input. Unfortunately, both the DUDC and RRC problems are also NP-complete
\cite{GJ79}, but unlike general set cover problem, DUDC and RRC problems admit a constant
factor approximation results. These two problems have been studied extensively due to their
wide applications in wireless networks \cite{CDDDFLNS10,FKKLS07,YMFXZ12}.

\subsection{Related work}
The DUDC problem has a long history in the literature. It is a NP-complete problem \cite{GJ79}.
The first constant factor approximation algorithm has been proposed by Br\"{o}nnimann and Goodrich
\cite{BG95}. Their algorithm is based on the epsilon nets concept. After that many authors proposed
constant factor approximation algorithm for the DUDC problem \cite{AEMN06,CDDDFLNS10,CMWZ04,CKL07,MR09,NV06}.
A summary of such results are available in \cite{DFLN11}. Using local search, Mustafa and Ray \cite{MR09}
proposed a PTAS for the DUDC problem. The time complexity of their PTAS is
$O(m^{2.(\frac{8\sqrt{2}}{\epsilon})^2+1}n)$ for $0 < \epsilon \leq 2$.
Therefore, the fastest operation of this algorithm is obtained when $\epsilon = 2$ for a 3-factor approximation
result in $O(m^{65}n)$ time, which is not practical even for $m = 2$. This leads to further research on the
DUDC problem for finding constant factor approximation algorithm with reasonable running time. Das et al. \cite{DFLN11}
proposed an 18-factor approximation algorithm. The running time of their algorithm is
$O(mn + n \log n + m\log m)$. Recently, Fraser and L\'{o}pez-Ortiz \cite{FL12} proposed an 15-factor approximation
algorithm for the DUDC problem, which runs in $O(m^6n)$ time. Das et al. \cite{DDN10} studied a restricted version
of the DUDC problem, where all the centers of disks in ${\cal D}$ are within a unit disk and all the points in ${\cal P}$
are outside of that unit disk. They proposed an 2-factor approximation algorithm for this restricted version of the
DUDC problem, which runs in $O((m+n)^2)$ time.

In the way to solve DUDC problem, some authors consider a restricted version of the DUDC problem.
In the literature it is known as {\it line-separable discrete unit disk cover} (LSDUDC) problem
\cite{CDDDFLNS10}. In this problem, the plane being divided into two half-planes $\ell^+$ and
$\ell^-$ defined by a line $\ell$, all the points in ${\cal P}$ are in $\ell^+$ and the centers of disks
in ${\cal D}$ are in $\ell^+ \cup \ell^-$ such that each point in ${\cal P}$ is covered by at least
one disk centered in $\ell^-$. Carmi et al. \cite{CKL07} described an 4-factor approximation algorithm
for the LSDUDC problem. Latter, Claude et al. \cite{CDDDFLNS10} proposed an 2-factor approximation algorithm
for LSDUDC problem. Another restricted version of the DUDC problem is {\it within strip discrete unit disk cover}
(WSDUDC) problem. In this problem, all the points in ${\cal P}$ and center of the disks in ${\cal D}$ are within
a strip of width $h$. Das et al. \cite{DFLN11} proposed an 6-factor approximation algorithm for $h = 1/\sqrt{2}$.
Latter, Fraser and L\'{o}pez-Ortiz \cite{FL12} proposed an $3 \lceil 1/{\sqrt{1-h^2}}\rceil$-factor approximation
result for $0 \leq h < 1$. They also proposed an 3-factor (resp. 4-factor) algorithm for $h \leq 4/5$
(resp. $h \leq 2 \sqrt{2}/3$).

Agarwal and Sharir \cite{AS98} studied Euclidean $k$-center problem. In this problem, a set ${\cal P}$ of $n$ points,
a set ${\cal O}$ of $m$ points, and an integer $k$ are given. The objective is to find $k$ disks centered at
the points in ${\cal O}$ such that each point in ${\cal P}$ is covered by at least one disk and the radius
of largest disk is minimum. This problem is known to be NP-hard \cite{AS98}. For fixed $k$, Hwang et al. \cite{HLC93}
presented a $m^{O(\sqrt{k})}$-time algorithm. Latter, Agarwal and Procopiuc \cite{AP02} presented $m^{O(k^{1-1/d})}$-time
algorithm for the $d$-dimensional points. Fowler et al. \cite{FPT81} considered the minimum geometric disk cover problem,
where input is a set ${\cal P}$ of points in the Euclidean plane and the objective is to compute minimum cardinality set
${\cal X}$ of unit disks such that each point in ${\cal P}$ is covered by at least one disks in ${\cal X}$. They proved
that the problem is NP-hard. Hochbaum and Maass \cite{HM85} proposed a PTAS for the geometric disk cover problem.

A {\it sector} is a maximal region formed by the intersection of a set of disks such that all points within the sector
are covered by the same set of disks. Funke et al. \cite{FKKLS07} proposed {\it greedy sector cover} algorithm for RRC
problem. The approximation factor of this algorithm is $O(\log w)$, where $w$ is the maximum number of sectors covered
by a single disk. They proved that the greedy sector cover algorithm has an approximation algorithm no better than
$\Omega(\log w)$. In the same paper, they proposed grid placement algorithm (based on the algorithm proposed by
Bose et al. \cite{BMMM01}) and proved that their algorithm produce $18\pi$-factor approximation
result. Though the algorithm is not guaranteeing full coverage of the region of interest, the area that
remains uncovered can be bounded by the number of chosen grids. In the same paper, they have also considered RRC problem
in different setup. We denote this setup as {\it reduced radius} setup. In the reduce radius setup, the region
of interest ${\cal R}$ is also covered by the disks in ${\cal D}$ after reducing their radius to $(1-\delta)$. $\delta$ is
said to be {\it reduce radius parameter}. The reduce radius setup has many applications in wireless sensor networks, where
coverage remains stable under small perturbations of sensing ranges/positions. In this setup an algorithm ${\cal A}$ is said
to be $\beta$-factor approximation if $\frac{|{\cal A}_{out}|}{|opt|} \leq \beta$, where ${\cal A}_{out}$ is the output of algorithm
${\cal A}$ and $opt$ is the optimum set of disks with reduced radius covering the region of interest. In reduce radius setup,
Funke et al. \cite{FKKLS07} proposed an 4-factor approximation algorithm for RRC problem.

\subsection{Our Results}
We provide a PTAS (i.e., $(1+\mu)$-factor approximation algorithm) for LSDUDC problem, which runs in
$O(m^{2(1+ \frac{1}{\mu})}n)$ time ($0 < \mu \leq 1$). Using this PTAS, we present an $(9+\epsilon)$-factor approximation 
algorithm for DUDC problem in $O(\max(m^{2(1+ \frac{6}{\epsilon})}n, m^6n)$ time, where $0 < \epsilon \leq 6$. For the 
RRC problem, we describe an $(9+\epsilon)$-factor approximation algorithm using the algorithm for DUDC. We also propose 
an 2.25-factor approximation algorithm in the reduce radius setup. The previous best known approximation factor was 4 \cite{FKKLS07}.

\section{PTAS for LSDUDC Problem}
Let $\ell$ be a horizontal line. We use $\ell^+$ (resp. $\ell^-$) to denote the half-plane
above (resp. below) the line $\ell$. For LSDUDC problem, we assume that each point in ${\cal P}$ is
in $\ell^-$ and center of each disk in ${\cal D}$ is in $\ell^+ \cup \ell^-$ such that union of the
disks centered in $\ell^+$ covers all points in ${\cal P}$. The objective is to find minimum cardinality
set ${\cal D}^* \subseteq {\cal D}$ such that each point in ${\cal P}$ is covered by at least one
disk in ${\cal D}^*$. Now, we want to define some terminology.

\begin{definition}
We use ${\cal U}$ (resp. ${\cal L}$) to denote the set of disks in ${\cal D}$ with centers in
$\ell^+$ (resp. $\ell^-$). We use $circ(d)$ (resp. $\alpha(d)$) to denote the boundary arc (resp. center)
of the disk $d$. A disk $d \in {\cal U}$ is said to be {\it lower boundary disk} if there does not exit
$X \subseteq {\cal U}\setminus \{d\}$ such that $d \cap \ell^- \subset (\cup_{D \in X} D) \cap \ell^-$.
For a lower boundary disk $d \in {\cal U}$, we use the term {\it lower region} to denote the region
$d \cap \ell^-$ and {\it lower arc} to denote the arc $circ(d) \cap \ell^-$ (see Fig. \ref{figure-1}).
We use $D_\ell = \{d^1, d^2, \ldots, d^s\} \subseteq {\cal U}$ to denote the set of all lower boundary
disks. We use $B_{region}$ to denote the region covered by the disks in
$D_\ell$.
\end{definition}

\begin{figure}[h]
\centering
\includegraphics[width=3in]{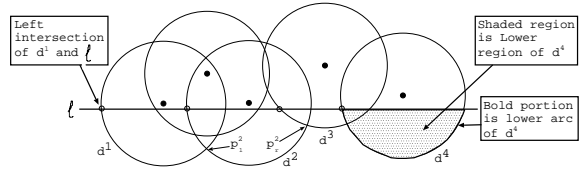}
\caption{Lower region, left intersection and lower boundary disks}
\label{figure-1}
\end{figure}

Needless to mention that each disk in $D_\ell$ intersect the horizontal line $\ell$. Without
loss of generality assume that $d^1, d^2, \ldots, d^s$ is the sorted order from left to
right based on their left intersection point with the line $\ell$ (see Fig. \ref{figure-1}).
Since centers of the disks in $D_\ell$ are in $\ell^+$ the number of intersection (if any) of two
disks of $D_\ell$ in $\ell^-$ is one. For each disk $d^i \in D_\ell$ we define two points, namely
$p^i_l$ and $p^i_r$ as follows:

\begin{description}
\item[$p^i_l$:] If the disk $d^i$ has intersection with $d^{i-1}$ in $\ell^-$, then $p^i_l$ is the
intersection point between $circ(d^{i-1})$ and $circ(d^i)$ in $\ell^-$, otherwise $p^i_l$ is the
left intersection point between $\ell$ and $circ(d^i)$.

\item[$p^i_r$:] If the disk $d^i$ has intersection with $d^{i+1}$ in $\ell^-$, then $p^i_r$ is the
intersection point between $circ(d^{i+1})$ and $circ(d^i)$ in $\ell^-$, otherwise $p^i_r$ is the
right intersection point between $\ell$ and $circ(d^i)$.
\end{description}

Where $d^0$ and $d^{s+1}$ are the two dummy disks having no intersection with $d^1$ and $d^s$ respectively.
For each $i = 1, 2, \ldots, s$ let ${\cal P}_i (\subseteq {\cal P})$ be the set of points lies between two
vertical lines through $p^i_l$ and $p^i_r$. Let $e^i$ be the vertical line through the point $p^i_r$
for $i = 1, 2, \ldots, s$. We use $e^{i-}$ (resp. $e^{i+})$ to denote the region in the left (resp. right)
side of the vertical line $e^i$. Let $D^{i-}$ (resp. $D^{i+})$ be the optimum cover of the points in
${\cal P} \cap e^{i-}$ (resp. ${\cal P} \cap e^{i+}$).

\begin{algorithm}
\caption{LSDUDC$({\cal P}, {\cal D}, k, \ell)$}
\begin{algorithmic}[1]
\STATE {\bf Input:} Set ${\cal P}$ of points, set ${\cal D}$ of unit disks, a positive integer $k$ and a
                    horizontal line $\ell$ such that ${\cal P} \cap \ell^- = {\cal P}$ and union of the
                    disks centered in $\ell^+$ covers all the points in ${\cal P}$.
\STATE {\bf Output:} Set ${\cal D}^* \subseteq {\cal D}$ of disks covering all the points in ${\cal P}$.\\

\STATE Set ${\cal D}^* \leftarrow \emptyset$
\STATE Find lower boundary disks set $D_\ell$ and arrange them from left to right as defined above. Let
$D_\ell = \{d^1, d^2, \ldots, d^s\}$ be the lower boundary disks from left to right.
\FOR {($i = 1, 2, \ldots s$)}
    \STATE Compute the set ${\cal P}_i (\subseteq {\cal P})$
\ENDFOR
\STATE Set $i \leftarrow 1$
\WHILE{($i \leq s$)}
\STATE Find the maximum index $j$ such that $\cup_{\theta=i, i+1, \ldots j} {\cal P}_\theta$ is covered by 
a set ${\cal D}_1 (\subseteq {\cal D})$ of size $k$.
\STATE ${\cal D}^* = {\cal D}^* \cup {\cal D}_1$, $i \leftarrow j+1$
\ENDWHILE
\STATE Return ${\cal D}^*$
\end{algorithmic}
\label{alg:LSDUDC}
\end{algorithm}

\begin{observation} \label{observation-1}
For two disks $d', d'' \in {\cal D}$, if $d', d'' \in D^{i-}$ and $d', d'' \in D^{i+}$,
then both the disks $d'$ and $d''$ intersect $e^i$ and $circ(d')$ and $circ(d'')$ intersect
in $B_{region}$.
\end{observation}

\begin{proof}
Both the disks $d'$ and $d''$ intersect $e^i$ because $d', d'' \in D^{i-}$ and $d', d'' \in D^{i+}$.
Now, if $circ(d')$ and $circ(d'')$ does not intersect in $B_{region}$, then either
${\cal P} \cap d' \subseteq {\cal P} \cap d''$ or ${\cal P} \cap d'' \subseteq {\cal P} \cap d'$.
Therefore, both $d'$ and $d''$ can not appear in the solutions $D^{i-}$ and $D^{i+}$. \hfill $\Box$
\end{proof}

\begin{definition}
A pair $(d', d'')$ of disks is said to be week (resp. strong) cover pair if $circ(d')$ and
$circ(d'')$ intersect once (resp. twice) in $B_{region}$.
\end{definition}

\begin{lemma} \label{lemma-1}
For a week cover pair $(d', d'')$; $d', d'' \in D^{i-}$ and $d', d'' \in D^{i+}$ can not
happen simultaneously.
\end{lemma}

\begin{proof}
On contrary, assume $d', d'' \in D^{i-}$ and $d', d'' \in D^{i+}$. By the definition of week cover pair,
$circ(d')$ and $circ(d'')$ does not intersect either in $e^{i-}$ or $e^{i+}$.
Therefore, either (i) ${\cal P}_{e^{i-}} \cap d' \subset {\cal P}_{e^{i-}} \cap d''$ or
${\cal P}_{e^{i-}} \cap d'' \subset {\cal P}_{e^{i-}} \cap d'$ or (ii) ${\cal P}_{e^{i+}} \cap d'
\subset {\cal P}_{e^{i+}} \cap d''$ or ${\cal P}_{e^{i+}} \cap d'' \subset {\cal P}_{e^{i+}} \cap d'$,
where ${\cal P}_{e^{i-}}$ (resp. ${\cal P}_{e^{i+}}$) is the set of points in ${\cal P}$ to the left
(resp. right) of $e^i$ (see Fig. \ref{figure-2}). Thus, both the disks $d', d''$ can not be in
$D^{i-}$ and $D^{i+}$.  \hfill $\Box$
\end{proof}

\begin{figure}[h]
\centering
\includegraphics[width=3in]{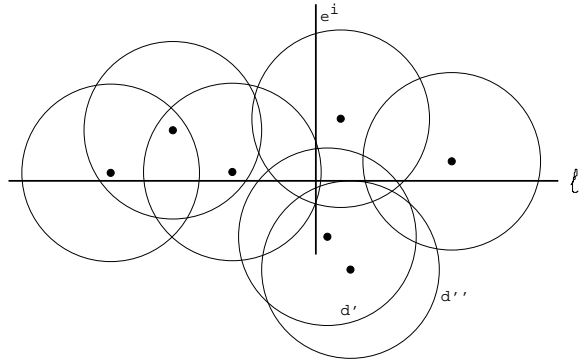}
\caption{Proof of Lemma \ref{lemma-1}}
\label{figure-2}
\end{figure}

\begin{lemma} \label{lemma-2}
For a strong cover pair $(d', d'')$, if $d', d'' \in D^{i-}$ and $d', d'' \in D^{i+}$, then one
intersection between $circ(d')$ and $circ(d'')$ lies in $e^{i-}$ and other intersection lies in $e^{i+}$.
\end{lemma}

\begin{proof}
If both the intersections between $circ(d')$ and $circ(d'')$ lies either in $e^{i-}$ or $e^{i+}$, then
from Observation \ref{observation-1} and the proof of Lemma \ref{lemma-1} $d', d'' \in D^{i-}$ and
$d', d'' \in D^{i+}$ can not happen simultaneously, which is a contradiction. Thus, the lemma follows.
\hfill $\Box$
\end{proof}

\begin{lemma}\label{lemma-3}
For a strong cover pair $(d', d'')$ with $\alpha(d')$ is above $\alpha(d'')$, if the intersection of
$circ(d')$ and $circ(d'')$ lies within a lower boundary disk $d$ and both $\alpha(d')$ and $\alpha(d'')$ lies
either right of $\alpha(d)$ or left of $\alpha(d)$, then one intersection of $circ(d)$ and
$circ(d')$ happen above the horizontal line $\ell$.
\end{lemma}

\begin{proof}
With out loss of generality assume that both $\alpha(d')$ and $\alpha(d'')$ lie to the right of $\alpha(d)$.
Let $a$ and $b$ be the two intersection points of $circ(d)$ and $circ(d')$. Since $\alpha(d')$ is above $\alpha(d'')$,
the intersection of $circ(d')$ and $circ(d'')$ lies within $d$ and both $\alpha(d')$ and $\alpha(d'')$ lie to the right
of $\alpha(d)$, $\alpha(d')$ should lies above at least one point among $a$ or $b$. Let $\alpha(d')$ lies above $a$.
By symmetry, $\overline{\alpha(d'),a}$ and $\overline{\alpha(d),b}$ are parallel (see Fig. \ref{figure-3}). Thus,
$b$ must be above $\alpha(d)$ i.e., $b$ must be above $\ell$. \hfill $\Box$
\end{proof}

\begin{figure}[h]
\centering
\includegraphics[width=3in]{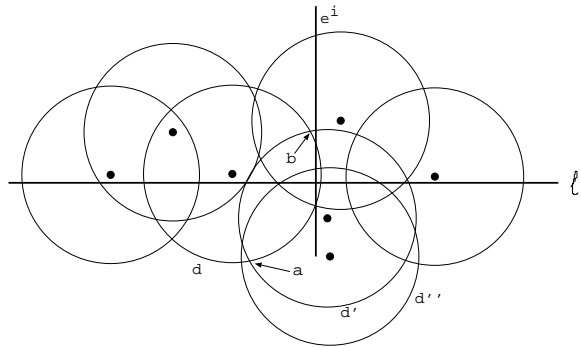}
\caption{Proof of Lemma \ref{lemma-3}}
\label{figure-3}
\end{figure}

\begin{lemma} \label{lemma-4}
$|D^{i-} \cap D^{i+}| \leq 2$.
\end{lemma}

\begin{proof}
On contrary, assume that $d_x, d_y, d_z \in D^{i-} \cap D^{i+}$. Since $d_x, d_y, d_z \in D^{i-}$ as well as
$d_x, d_y, d_z \in D^{i+}$, the disks $d_x, d_y, d_z$ intersect each other in $B_{region}$. If any pair $(d, d') \in 
\Gamma = \{(d_x, d_y), (d_x, d_z), (d_y,d_z)\}$ do not form a week cover pair nor strong cover pair, then either 
$d \cap B_{region} \subseteq d' \cap B_{region}$ or $d \cap B_{region} \supseteq d' \cap B_{region}$, which contradict 
the fact that $d_x, d_y, d_z \in D^{i-} \cap D^{i+}$. Again, from Lemma \ref{lemma-1}, no pair in $\Gamma$ form a week 
cover pair because $d_x, d_y, d_z \in D^{i-}$ as well as $d_x, d_y, d_z \in D^{i+}$. Therefore, each pair in $\Gamma$ 
form a strong cover pair. Without loss of generality assume that $\alpha(d_x)$ is below $\alpha(d_y)$ and $\alpha(d_y)$ 
is below $\alpha(d_z)$ (see Fig. \ref{figure-4}). If $a$ is the intersection between $circ(d_x)$ and $circ(d_y)$ inside 
the lower boundary disk $d$ (say) and below the horizontal line $\ell$, then one intersection between $circ(d_y)$ and 
$circ(d_z)$ lies inside of $d$ (from Lemma \ref{lemma-3}). Therefore, $(d_x \cup d_y \cup d_z)\cap {\cal P} \cap e^{i-} 
\subseteq (d_x \cup d)\cap {\cal P} \cap e^{i-}$, which implies that $D^{i-}$ is not optimum, leading to a contradiction. 
Thus, the lemma follows. \hfill $\Box$
\end{proof}

\begin{figure}[h]
\centering
\includegraphics[width=3in]{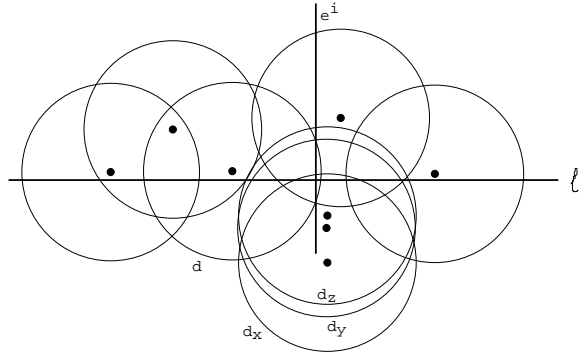}
\caption{Proof of Lemma \ref{lemma-4}}
\label{figure-4}
\end{figure}

The following theorem says that the LSDUDC problem admits a PTAS.

\begin{theorem} \label{theorem-1}
Algorithm \ref{alg:LSDUDC} produce $(1+\frac{2}{k-2})$-factor approximation results in $O(m^kn)$ time.
\end{theorem}

\begin{proof}
For some integer $t$, let $j_1, j_2, \ldots, j_t$ the values of $j$ in the while loop (line number 9) of
the Algorithm \ref{alg:LSDUDC}. Let ${\cal Q}_{v} =
\cup_{i=j_{v-1}+1, j_{v-1}+2, \ldots, j_v}{\cal P}_i$ for $v = 1, 2, \ldots, t$, where $j_0 = 0$.
Algorithm \ref{alg:LSDUDC} finds a covering for the sets $\{{\cal Q}_1, {\cal Q}_2, \ldots, {\cal Q}_t\}$
independently with each of size $k$ (optimum size because in each iteration of the while loop in line
number 9, Algorithm \ref{alg:LSDUDC} finds maximum value of $j$'s) except the covering of ${\cal Q}_t$. Let
${\cal D}^1, {\cal D}^2, \ldots, {\cal D}^t$ be the covering for ${\cal Q}_1, {\cal Q}_2, \ldots, {\cal Q}_t$
respectively. Lemma \ref{lemma-4} says that ${\cal D}^i \cap {\cal D}^{i+1} \leq 2$. Therefore, the total
number of disks required to cover all the points by Algorithm \ref{alg:LSDUDC} is $k(t-1) + |{\cal D}^t|$
whereas at least $(k-2)(t-1) + |{\cal D}^t|$ disks required in the optimum solution. Thus, the approximation
factor of the Algorithm \ref{alg:LSDUDC} is $(1+\frac{2}{k-2})$.

The execution time to find lower boundary disks and arrange them from left to right (line number 4) is $O(m \log m)$,
where $m = |{\cal D}|$. To compute ${\cal P}_i$ for $i = 1, 2, \ldots s$ (for loop in line number 5) $O(n \log n)$ time
is required. The time complexity of the while loop (line number 9) is $O(m^kn)$. Thus, the total time complexity of
the Algorithm \ref{alg:LSDUDC} is $O(m^kn)$. \hfill $\Box$
\end{proof}

\subsection{An $(9+\epsilon)$-factor Approximation Algorithm for DUDC Problem}
In this section, we wish to describe an $(9+\epsilon)$-factor approximation algorithm for DUDC Problem,
where a set ${\cal P}$ of $n$ points and a set ${\cal D}$ of $m$ unit disks are distributed in the Euclidean
plane and objective is to choose minimum cardinality set ${\cal D}^* (\subseteq {\cal D})$ such that union of
the disks in ${\cal D}^*$ covers ${\cal P}$. From Theorem \ref{theorem-1}, LSDUDC problem has an
$(1+\mu)$-factor approximation algorithm ($\mu = \frac{2}{k-2}$) and the running time of the
algorithm is $O(m^{2(1+ \frac{1}{\mu})}n)$. Das et al. \cite{DFLN11} proved that any instance of DUDC
problem can be partitioned into several instances of LSDUDC and WSDUDC (with strip width $1/\sqrt{2}$) problems.
They also proved that the approximation factor of the DUDC problem is

$6 \times$ (approximation factor of LSDUDC problem) $+$ \\
approximation factor of WSDUDC (width $\geq 1/\sqrt{2}$) problem.

Fraser and L\'{o}pez-Ortiz \cite{FL12} proposed an 3-factor approximation algorithm for WSDUDC (with width $h \leq 4/5$)
problem. Therefore, we have the following theorem for DUDC problem.

\begin{theorem}\label{theorem-2}
The DUDC problem admits $(9+\epsilon)$-factor approximation result in $O(\max(m^6n, m^{2(1+ \frac{6}{\epsilon})}n)$
time.
\end{theorem}

\begin{proof}
The approximation factor of the DUDC problem is $(6 \times$ (approximation
factor of LSDUDC problem) $+$ approximation factor of WSDUDC (width $\geq 1/\sqrt{2}$) problem) \cite{DFLN11}. Therefore,
the approximation factor for the DUDC problem is $6 \times (1+\mu) + 3 = 9 + \epsilon$, where $\epsilon = 6 \mu$.
The time complexity result follows from (i) time complexity of WSDUDC is $O(m^6n)$ \cite{FL12} and (ii) above
discussion. \hfill $\Box$
\end{proof}

\begin{corollary}
The DUDC problem has an 12-factor approximation result with time complexity equal to the 3-factor approximation
algorithm for WSDUDC problem.
\end{corollary}

\begin{proof}
The time complexity of 3-factor approximation algorithm for the WSDUDC problem is $O(m^6n)$ \cite{FL12}. If we
set $\epsilon = 3$, then approximation factor of the DUDC problem is 12 (see Theorem \ref{theorem-2}) and
the running time of the algorithm is $O(m^6n)$ (see Theorem \ref{theorem-2}). Thus, the corollary follows.
\hfill $\Box$
\end{proof}

\section{Approximation Algorithms for RRC Problem}
In the RRC problem, the inputs are (i) a set ${\cal D}$ of $m$ unit disks and (ii) a rectangular region ${\cal R}$
and the objective is to choose minimum cardinality set ${\cal D}^{**} (\subseteq {\cal D})$ such that each point
in ${\cal R}$ is covered by at least one disk in ${\cal D}^{**}$. In this problem, we assume that ${\cal R}$ is
covered by the union of the disks in ${\cal D}$, otherwise the RRC problem has no feasible solution. A {\it sector}
$f$ inside ${\cal R}$ is a maximal region inside ${\cal R}$ formed by the intersection of a set of disks such that
each point within $f$ is covered by the same set of disks. Let ${\cal F}$ be the set of all sectors (inside ${\cal R}$)
formed by ${\cal D}$. Therefore, the size of ${\cal F}$ is at most $O(m^2)$. Now we construct a set of points ${\cal T}$
as follows: for each sector $f \in {\cal F}$ we add one arbitrary point $p \in f$ to ${\cal T}$. Thus, we have the 
following theorem:

\begin{theorem}\label{theorem-3}
The RRC problem has $(9+\epsilon)$-factor approximation algorithm with running time
$O(\max(m^8, m^{4(1+3/\epsilon)})$.
\end{theorem}

\begin{proof}
Consider an arbitrary point $p \in {\cal T}$. Let $f \in {\cal F}$ be the sector in which the point $p$ lies. From the
definition of sector, if a disk $d \in {\cal D}$ covers $p$, then the disk $d$ also covers the whole sector $f$. Therefore,
the instance $({\cal R}, {\cal D})$ of the RRC problem is exactly same as the instance $({\cal T}, {\cal D})$ of the
DUDC problem. Note that ${\cal T} = O(m^2)$. Thus, the theorem follows from Theorem \ref{theorem-2} by putting $n = m^2$.
\hfill $\Box$
\end{proof}

\subsection{RRC Problem in Reduce Radius Setup}
In this subsection we consider the RRC problem in reduce radius setup. For a given set ${\cal D}$ of unit disks and
a rectangular region ${\cal R}$ such that ${\cal R}$ is covered by the disks in ${\cal D}$ after reducing their
radius to $(1-\delta)$, the objective is to choose minimum cardinality set ${\cal D}^{**} (\subseteq {\cal D})$ whose
union covers ${\cal R}$. In the reduce radius setup an algorithm ${\cal A}$ is said to be $\beta$-factor approximation if
$\frac{|{\cal A}_{out}|}{|opt|} \leq \beta$, where ${\cal A}_{out}$ is the output of ${\cal A}$ and $opt$ is the
optimum set of disks with reduced radius covering the region of interest. The reduce radius setup has many applications 
in wireless sensor networks, where coverage remains stable under small perturbations of sensing ranges and their positions. 
Here we propose an 2.25-factor approximation algorithm for this problem. The best known approximation factor for the same 
problem was 4 \cite{FKKLS07}.

\begin{observation} \label{observation-2}
Let $\delta = \nu/\sqrt{2}$ and $d$ be an unit disk centered at a point $p$. If $d'$ is a disk of radius
$(1-\delta)$ centered within a square of size $\nu \times \nu$ centered at $p$, then $d' \subseteq d$.
\end{observation}

\begin{proof}
The observation follows from the fact that the maximum distance of any point within the square
of size $\nu \times \nu$ from the center point $p$ is $\nu/\sqrt{2} (=\delta)$. \hfill $\Box$
\end{proof}

Consider a grid with cells of size $\nu \times \nu$ over the region ${\cal R}$. Like Funke et al. \cite{FKKLS07}
we also snap the center of each $d \in {\cal D}$ to the closest vertex of the grid and set its radius to
$(1-\delta)$. Let ${\cal D}'$ be the set of disks with radius $(1-\delta)$ after snapping their centers.
Let ${\cal R}'$ be a square of size $4 \times 4$ on the Euclidean plane contained in ${\cal R}$. We define the
regions {\it TOP, DOWN, LEFT, RIGHT, TOP-LEFT, TOP-RIGHT, DOWN-LEFT, DOWN-RIGHT} as shown in Fig. \ref{figure-5}.
We now construct a set ${\cal D}_{RS} (\subseteq {\cal D}')$ such that no disk $d \not \in {\cal D}_{RS}$ can
participate to the optimum solution for covering the region ${\cal R}'$ by ${\cal D}'$. The pseudo code
for construction of ${\cal D}_{RS}$ is given in Algorithm \ref{alg:DRS}.

\begin{figure}[h]
\centering
\includegraphics[width=1.5in]{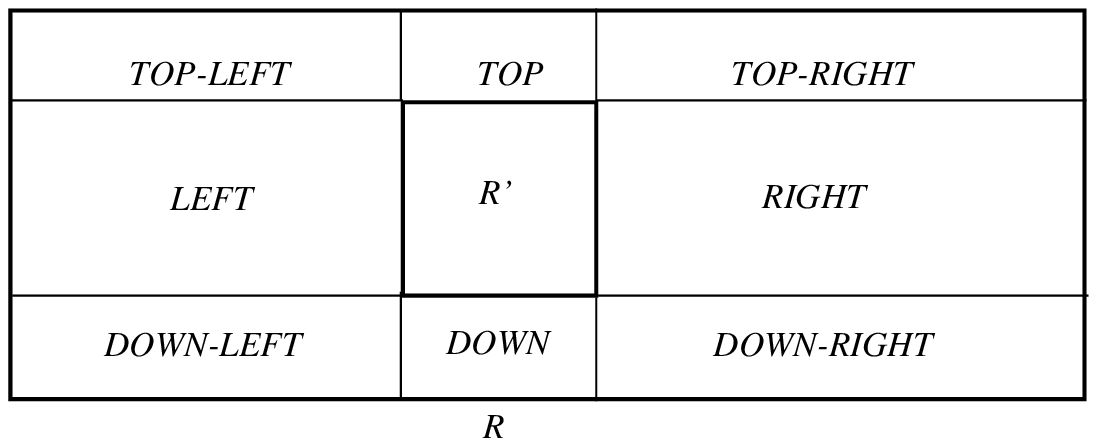}
\caption{Definition of different regions}
\label{figure-5}
\end{figure}

\begin{definition}
A disk $d \in {\cal D}'$ dominates another disk $d' \in {\cal D}'$ with respect to the region ${\cal R}'$ if
$d \cap {\cal R}' \supseteq d' \cap {\cal R}'$.
\end{definition}

\begin{algorithm}
\caption{$Algorithm\_{\cal D}_{RS}({\cal D}', {\cal R}', \nu)$}
\begin{algorithmic}[1]
\STATE {\bf Input:} Set ${\cal D}'$ of disks, a square region ${\cal R}'$ of size $4 \times 4$ and grid size $\nu$. \\
\STATE {\bf Output:} ${\cal D}_{RS} (\subseteq {\cal D}')$ such that no disk $d \not \in {\cal D}_{RS}$ can participate
                        to the optimum solution for covering the region ${\cal R}'$ by ${\cal D}'$. \\

\STATE Set ${\cal D}_{RS} \leftarrow \emptyset, {\cal D}_t \leftarrow \emptyset$
\STATE For each disk $d \in {\cal D}'$ having center in ${\cal R}'$, ${\cal D}_{RS} = {\cal D}_{RS} \cup \{d\}$
\STATE For each horizontal grid line segment $h$ in {\it LEFT} add a disk $d \in {\cal D}'$ to ${\cal D}_{RS}$
if (i) $d \cap {\cal R}' \neq \emptyset$, (ii) center of $d$ lies on $h$ and (iii) center of $d$ closest to
${\cal R}'$ than other disks having center on $h$. Similarly add disks to ${\cal D}_{RS}$ for the
regions {\it RIGHT, TOP} and {\it DOWN}.

\FOR{(each horizontal grid line segment $h$ in {\it TOP-RIGHT} from bottom to top)}
\STATE Add a disk $d \in {\cal D}'$ to ${\cal D}_t$ if (i) $d \cap {\cal R}' \neq \emptyset$, (ii) center of $d$
lies on $h$ and (iii) there does not exits any disk $d' \in {\cal D}_t$ dominating $d$.
\ENDFOR
\STATE ${\cal D}_{RS} = {\cal D}_{RS} \cup {\cal D}_t$
\STATE repeat steps 6-9 for {\it TOP-LEFT}, {\it DOWN-LEFT} and {\it DOWN-RIGHT}.
\STATE Return ${\cal D}_{RS}$
\end{algorithmic}
\label{alg:DRS}
\end{algorithm}

\begin{lemma}\label{lemma-5}
If $d \in {\cal D}'$ and $d \not \in {\cal D}_{RS}$, then $d$ can not participate to the optimum solution for covering
${\cal R}'$ by minimum number of disks in ${\cal D}'$.
\end{lemma}

\begin{proof}
Since $d \not \in {\cal D}_{RS}$ the center of $d$ is in outside of ${\cal R}'$ (see line number 4 of Algorithm \ref{alg:DRS}).
With out loss of generality assume that center of $d$ is in {\it LEFT} and on the horizontal grid line segment $h$. By our
construction of the set ${\cal D}_{RS}$, there exists a disk $d' \in {\cal D}_{RS}$ centered on $h$ such that
(a) $d' \cap {\cal R}' \neq \emptyset$, (b) center of $d'$ lies on $h$ and (c) center of $d'$ closest to ${\cal R}'$ than other
disks having center on $h$. Therefore, $d'$ dominates $d$. Similarly, we can prove for other cases also. Thus, the lemma follows.
\hfill $\Box$
\end{proof}

\begin{lemma}\label{lemma-6}
$|{\cal D}_{RS}| \leq \frac{16}{\nu^2}+\frac{20}{\nu}$.
\end{lemma}

\begin{proof}
The lemma follows from the following facts: (i) the maximum number of grid vertices in ${\cal R}'$ is $\frac{16}{\nu^2}$ and 
each of them can contribute one disk in ${\cal D}_{RS}$, (ii) the maximum number of horizontal grid line segment in the regions 
{\it TOP-LEFT, LEFT, DOWN-LEFT, DOWN-RIGHT, RIGHT} and {\it TOP-RIGHT} that can contribute a disk in ${\cal D}_{RS}$ is 
$\frac{12}{\nu}$ and (iii) the maximum number of vertical grid line segment in the regions {\it TOP} and {\it DOWN} that can 
contribute a disk in ${\cal D}_{RS}$ is $\frac{8}{\nu}$. Thus, the lemma follows. \hfil $\Box$
\end{proof}

From Observation \ref{observation-2} and Lemma \ref{lemma-6}, we can compute a cover of ${\cal R}'$ by ${\cal D}'' (\subseteq {\cal D})$
with minimum number of disks using brute-force method, where ${\cal D}''$ is the set of disks of unit radius corresponding to the disks in
${\cal D}_{RS}$. The running time of the brute-force algorithm is $O(2^{\frac{16}{\nu^2}+\frac{20}{\nu}})$ (see Lemma \ref{lemma-6}).
Though, the worst-case running time of this brute-force algorithm is exponential in $\frac{1}{\nu^2}$, in practice, it is very small. 
Note that time complexity of our proposed algorithm is less than that of 4-factor approximation algorithm proposed by Funke et al. 
\cite{FKKLS07}. We now describe approximation factor of our propose algorithm for RRC problem in reduce radius setup.

\begin{theorem}
In the reduce radius setup, RRC problem has an 2.25-factor approximation algorithm.
\end{theorem}

\begin{proof}
From the above discussion, for rectangle of size $4 \times 4$, we have optimum solution for RRC problem. Note that the
diameter of each disk of the RRC instance is 2. Therefore, we can apply shifting strategy described by Hochbaum and Maass \cite{HM85}
to solve RRC problem and the approximation factor is $(1+1/2)^2 = 2.25$. Thus, the theorem follows.
\hfil $\Box$
\end{proof}

\section{Conclusion}
In this paper, we have proposed a PTAS for LSDUDC problem, improving previous 2-factor approximation result. Using this PTAS,
we proposed an $(9+\epsilon)$-factor approximation algorithm for DUDC problem, improving previous 15-factor approximation result
for the same problem. The running time of our proposed algorithm for $\epsilon = 3$ (i.e., approximation factor of DUDC problem
is 12)is same as the running time of 15-factor approximation algorithm. We have also proposed an $(9+\epsilon)$-factor
approximation algorithm for RRC problem, which runs in $O(\max(m^8, m^{4(1+3/\epsilon)})$ time. In the reduce radius setup, we 
proposed an 2.25-factor approximation algorithm. The previous best known approximation factor was 4 \cite{FKKLS07}. The running 
time of our proposed algorithm for RRC problem in reduce radius setup is less than that of 4-factor approximation algorithm 
proposed in \cite{FKKLS07} for reasonably small values of $\delta (=\frac{\nu}{\sqrt{2}})$, where $\delta$ is the radius reduction 
parameter. Since the number of disks participating in the solution of $4 \times 4$ square is constant for fixed value of $\delta$, 
the number of disks participating in the solution of $L \times L$ square is constant. Therefore, using the shifting strategy proposed 
by Hochbaum and W. Maass \cite{HM85}, we can design a PTAS for the RRC problem in reduce radius setup.

\end{document}